\documentclass[aps,pra,superscriptaddress,twocolumn,a4paper,notitlepage]{revtex4-1}  
\usepackage{etex}
\usepackage{amsmath,amssymb,amsthm}
\usepackage[colorlinks=true,citecolor=blue,urlcolor=blue]{hyperref}
\usepackage[pdftex]{graphicx}
\usepackage{times,txfonts}
\usepackage{braket}
\usepackage{color}
\usepackage{natbib}
\usepackage{amsmath,blkarray}
\usepackage{mathtools}
\usepackage{latexsym}
\usepackage{tabularx, booktabs}
\usepackage{graphics,epstopdf}
\usepackage{graphicx}
\usepackage{float}
\usepackage{graphicx}
\usepackage{amsfonts}
\usepackage{subcaption}
\usepackage{color,soul}
\usepackage{graphics}
\usepackage{helvet}
\usepackage{dcolumn}
\usepackage{bm}
\usepackage{amsmath}
\usepackage{color}
\usepackage{times,txfonts}
\usepackage{amssymb}
\usepackage{amsfonts}
\usepackage{graphicx}
\usepackage{physics}
\usepackage{soul}
\usepackage{amsthm}
\usepackage{placeins}
\usepackage{soul}

\newtheorem{theorem}{Theorem}
\newtheorem{definition}{Definition}

\numberwithin{tsubcase}{tcase}

\newcommand{\tibc}{2-{\rm IBC}}
\newcommand{\thibc}{3-{\rm IBC}}
\newcommand{\nbc}{{\rm IBC}}

\newcommand{\Bb}{\mathcal{B}}
\newcommand{\id}{\mathbf{I}} 

\def\ben{\begin{equation*}}
\def\een{\end{equation*}}
\def\be{\begin{equation}}
\def\ee{\end{equation}}

\begin{document}
\allowdisplaybreaks 
\addtolength{\jot}{1em} 

\title{ Interplay of nonlocality and incompatibility breaking  qubit channels}
	\author{Swati Kumari}
\email{swatipandey084@gmail.com}
\affiliation{Department of Physics and Center for Quantum Frontiers of Research \& Technology (QFort), National Cheng Kung University, Tainan 701, Taiwan
 }
\author{Javid Naikoo}
\email{j.naikoo@cent.uw.edu.pl}
\affiliation{Centre for Quantum Optical Technologies, Centre of New Technologies, University of Warsaw, Banacha 2c, 02-097 Warsaw, Poland}
\author{Sibasish Ghosh}
\email{sibasish@imsc.res.in}
\affiliation{Optics and Quantum Information Group, The Institute of Mathematical Sciences, HNBI, CIT Campus, Taramani, Chennai 600113, India}
\author{A. K. Pan}
\email{akp@phy.iith.ac.in}
\affiliation{Department of physics, Indian Institute of Technology Hyderabad Kandi,Telengana, India}	
\begin{abstract}
Incompatibility and nonlocality are not only of foundational interest but also act as important resources for quantum information theory. In the  Clauser-Horne-Shimony-Holt (CHSH) scenario, the incompatibility of a
pair of observables is known to be equivalent to Bell nonlocality. Here, we investigate these notions in the context of qubit channels. The Bell-CHSH inequality has a greater perspective---compared to any genuine
tripartite nonlocality scenario---while determining the interplay between nonlocality breaking qubit channels
and incompatibility breaking qubit channels. In the Bell-CHSH scenario, we prove that if the conjugate of a channel is incompatibility breaking, then the channel is itself nonlocality breaking and vice versa. However, this
equivalence is not straightforwardly generalized to multipartite systems, due to the absence of an equivalence relation between incompatibility and nonlocality in the multipartite scenario. We investigate this relation in the tripartite scenario by considering some well-known states like Greenberger-Horne-Zeilinger and $W$ states and using the notion of Mermin and Svetlichny nonlocality. By subjecting the parties in question to unital
qubit channels, we identify the range of state and channel parameters for which incompatibility coexists with nonlocality. Further, we identify the set of unital qubit channels that is Mermin or Svetlichny nonlocality breaking
irrespective of the input state.
\end{abstract}
\maketitle
\section{Introduction}
 Nonlocality is one of the most profound notions in quantum mechanics ~\cite{bell1964einstein} and is often discussed  in conjunction with incompatibility of observables. Recent developments in quantum information theory have found nonlocality  a useful phenomenon underpinning many advantages afforded by various quantum information  processing tasks ~\cite{Brunner2014bell}. Nonlocality can also be considered as a potential quantum resource for information processing, such as in developing quantum protocols to reduce the amount of communication needed in certain computational tasks  \cite{Brunner2014bell} and providing secure quantum communications  \cite{scarani2001,ekart1991}. Incompatibility, like nonlocality, is not merely of theoretical interest but of practical utility, for example, in order to explore the advantage of entanglement shared by two parties in a cryptography  task, each party needs to carry out measurements that are  incompatible, in the sense that these cannot be  carried out simultaneously by a single measurement device. Incompatibility should not be confused with noncommutativity or the related concept of uncertainty principle. The notion of incompatibility is best understood in terms of joint measurability \cite{heinosaari2016invitation}. A collection of quantum measurements is jointly measurable, if it can be simulated by a single common quantum measurement device. If such a single common device cannot be constructed by a given set of quantum measurements, it then enables the set to be used as a quantum resource. This was first noted in \cite{wolf2009measurements} in the context of Clauser-Horne-Shimony-Holt  (CHSH) inequalities and later in the Einstein-Podolsky-Rosen steering, which is more explicit, when incompatibility appears as a quantum resource. Incompatibility is necessary and sufficient for the violation of the steering inequalities \cite{quintino,ola}. The relation between incompatibility and contextuality has also been studied in references \cite{kochen1975problem,budroni2021quantum}.  Further, a set of observables that is pairwise incompatible, but not triplewise can violate the Liang-Spekkens-Wiseman noncontextuality inequality \cite{LIANG20111}. Recently, the connection between steerability and measurement incompatibility was studied in \cite{huannature22} in the context of the so-called steerability equivalent observables. Thus, both nonlocality and incompatibility can be considered as quantum resources whose understanding is of utmost importance in view of emerging quantum technologies.

The interplay of nonlocality and incompatibility has been a subject matter of various studies. It is well known that any incompatible local measurements, performed by  the constituent parties of a system, lead to the violation of Bell inequality provided they share a pure entangled state ~\cite{bell1964einstein,Brunner2014bell}. Absence of either of them (i.e., entanglement or  incompatibility) will not allow the system to exhibit nonlocality. It is important to mention here that the  notion of quantum nonlocality without entanglement has been proposed in~\cite{Bennett1999} which is different from Bell nonlocality~ \cite{bell1964einstein} and amounts to the inability of discriminating a set of product states by local operations and classical communication, while mutual orthogonality of the states assures their perfect global discrimination.
 
  Further, for any  pair of dichotomic incompatible observables, there  always exists an entangled state which enables the violation of a Bell inequality  ~\cite{wolf2009measurements}. The relationship of incompatibility and nonlocality is sensitive to the dimension of the system;  for example, increasing the dimension beyond $2$, the incompatible observables do not necessarily lead to the violation of Bell-type inequalities,  implying that the measurement incompatibility can not guarantee nonlocality in general  ~\cite{bene2018measurement,hirsch2018quantum}. Here, we probe the interplay between incompatibility and nonlocality in the tripartite case by using the well-known Mermin and Svetlichny inequalities ~\cite{Svetlichny1987}. The Svetlichny inequality, unlike the Mermin inequality, is a genuine measure of nonlocality that assumes nonlocal correlations between two parties which are locally related to a third party and is known to provide a suitable measure to detect tripartite nonlocality for $W$ and Greenberger-Horne-Zeilinger (GHZ) classes of states \cite{Rungta2010}. We refer the interested reader to  \cite{Bancal2013,Brunner2014bell} for various facets of the multipartite nonlocality.

The extent to which a system can exhibit nonlocal correlations is also sensitive to its interaction with the  ambient environment. Such interaction is usually accompanied with a depletion of various quantum features like coherence, entanglement and nonlocality. The reduced dynamics of the system in such cases is given by completely positive and trace  preserving maps, also known as quantum channels (QCs). On the other hand, the action of conjugate channels on projective measurements turns them into unsharp positive operator-valued measures (POVMs) which may be biased,  in general. In light of the above discussion, a study of the  open system effects on the interplay of nonlocality and incompatibility naturally leads to the notions of  nonlocality breaking  and incompatibility breaking   quantum channels  ~\cite{pal2015non, heinosaari2016invitation}. A nonlocality breaking channel (NBC) can be defined as a channel which when applied to a system (or part of it) leads to a state which is local \cite{pal2015non}, while the  incompatibility breaking channel (IBC) is the one that turns incompatible  observables into compatible ones \cite{heinosaari2015,heinosaari2015incompatibility}. An IBC that renders any set of $n$ incompatible observables  compatible would be denoted by $n$-IBC. The notion of the NBC has been introduced in a  similar spirit of  well-studied entanglement breaking channels \cite{horodecki10}. Every entanglement breaking channel is nonlocality breaking but the converse is not true. As an example, the  qubit depolarizing channel  $ \mathcal{E}\big(\rho\big) :=p\big(\rho\big)+(1-p)\mathbf{I}/2$ is CHSH nonlocality breaking  for all  $\frac{1}{3} \leq p \leq \frac{1}{2}$, but not entanglement breaking \cite{pal2015non}. Hence, based on this classification, nonlocality and entanglement emerge as different resources.

  The equivalence of the steerability breaking channels and the incompatibility breaking channels was reported in \cite{Kiukas2017} and CHSH nonlocality breaking channels were shown to be a strict subset of the steerability breaking channels  \cite{huanyu2022}.    The connection between Bell nonlocality and incompatibility of two observable is well understood, however, the question of the equivalence between NBC ~\cite {pal2015non} and  IBC ~\cite {heinosaari2015incompatibility} is rarely discussed. This motivates us to  explore the relation between CHSH nonlocality breaking channels (CHSH-NBC) and $2$-IBC,  such that the action of one may be replaced by the other. The tripartite nonlocality has much richer and complex structure and less  is known about its synergy with incompatibility as compared to its bipartite counterpart. Mermin inequality assumes local-realistic correlations among all  the three qubits; hence a violation would be a signature of the tripartite nonlocality shared among the qubits. However, biseparable states were shown to also violate the Mermin inequality  \cite{Collins2022, scarani2001}. This motivated  Svetlichny  to introduce the notion of genuine tripartite nonlocality \cite{Svetlichny1987} and provide a set of inequalities sufficient to witness it. We make use of these notions of   \emph{absolute} and  \emph{genuine} nonlocality  to figure out the ranges of state and channels parameters in which NBC and $2$-IBC  coexist. 

The paper is organized as follows. In Sec.~ \ref{Preliminaries}, we revisit some basic notions and definitions used in this paper. Section~\ref{sec:results} is devoted to results and their discussion where we  prove an equivalence between NBCs and $2$-IBCs in the CHSH scenario. This is followed by an analysis of these notions in the tripartite scenario, where we identify the state and channel parameters in which NBCs and 2-IBCs co-exist. A conclusion is given in  Sec.~\ref{sec:Concl}.
\section{Preliminaries} \label{Preliminaries}
In this section, we discuss the notion of incompatibility in the context of observables and quantum channels and look at specific cases of bipartite and tripartite scenarios. 

		
\subsection{Incompatibility}
\subsubsection{Incompatibility of observables}
  A finite collection of observables $A_1, \dots, A_n$ associated with the respective outcome spaces $\Omega_{A_1},  \dots, \Omega_{A_n}$, is said to be \textit{compatible} (or \textit{jointly measurable}) if there exists a \textit{joint observable} $G$, defined over the product outcome space $\Omega_{A_1} \times \cdots \times \Omega_{A_n}$, such that for all $X_1 \subset \Omega_{A_1}, \dots, X_n \subset \Omega_{A_n} $, the following marginal relations hold \cite{pbusch}: 
\begin{align}
	\sum_{a_i, i=1,\dots,n;  i \ne k} G(a_1, \dots, a_n)    =  A_k (a_k)
\end{align}
where $a_k$ is the outcome associated to observable $A_k$ and the summation is carried out over all outcomes $a_i$ except for $i = k$. The notion of the incompatibility of observables can be illustrated by a simple example  of Pauli matrices $\sigma_x$ and $\sigma_z$ which are noncommuting and can not be measured jointly. However, consider the unsharp observables $S_x(\pm)= \frac{1}{2}(\mathbf{I} \pm  \frac{1}{\sqrt{2}} \sigma_x)$ and $S_z(\pm) = \frac{1}{2}(\mathbf{I} \pm  \frac{1}{\sqrt{2}} \sigma_z)$, with  $\mathbf{I}$ being the $2 \times 2$ identity matrix. The joint observable $G(i,j) = \frac{1}{4} (\mathbf{I} + \frac{i}{\sqrt{2}} \sigma_x + \frac{j}{\sqrt{2}} \sigma_z)$ with $i,j = \pm 1$ jointly determines the probabilities of generalized measurements $S_x$ and $S_z$, since the later can be obtained as marginals $S_x(\pm) = \sum_{j} G(\pm, j)$ and $S_z(\pm) = \sum_{i} G(i,\pm)$.\bigskip

\subsubsection{Incompatibility breaking quantum channel}
A QC, in the Schr{\"o}dinger picture, is a completely positive trace preserving   map $\mathcal{E}:\mathcal{L}({\mathcal{H^A}}) \rightarrow {\mathcal{L}}({\mathcal{H^B}})$,   where ${\mathcal{L}}({\mathcal{H}^i})$ is the set of bounded linear operators on Hilbert space $\mathcal{H}^i$ ($i=A, B$).  
One may write the operator sum representation~\cite{nielsenchuangbook}
\begin{align}
	\label{Cha}
	\rho^\prime =\mathcal{E}\big(\rho\big) =\sum_{j=1}^{n}K_j\rho K_j^\dagger
\end{align}
where $K_i$ are known as Kraus operators satisfying the completeness relation $\sum_j K^\dag_j K_j=\id$. The QCs which map the identity
operator to itself, i.e., $\mathcal{E}(\id)=\id$, are  known as unital QCs.  A quantum channel $\mathcal{E}$ in the Schr{\"o}dinger picture acting on quantum state $\rho$ can be thought of as \textit{conjugate} channel $\mathcal{E}^*$ acting on observable $A$ through the following duality relation 
\begin{equation}\label{eq:SHpictures}
	\Tr \left[ \mathcal{E}(\rho) A \right] = \Tr \left[\rho~ \mathcal{E}^*(A) \right]. 
\end{equation}
\begin{definition}
A quantum channel $\mathcal{E} $ is said to be \textit{incompatibility breaking} if the outputs $\mathcal{E}^*(A_1), \dots, \mathcal{E}^*(A_n)$ are compatible for any choice of input observables $A_1, \dots, A_n$ $(n\geq2)$. 
\end{definition} 
If a channel $\mathcal{E}$ breaks the incompatibility of every class of $n$ observables, it is said to be $n -$ IBC. For example, a channel would be  $\tibc$  if it  breaks the incompatibility of a pair of observables. As an example, the white noise mixing channel $\mathcal{W}_{\eta}$  is described as 
\begin{equation}\label{eq:Wt}
	\rho^\prime = \mathcal{W}_{\eta} [\rho] = \eta \rho + (1-\eta) \frac{\mathbf{I}}{d}.
\end{equation}
Here, $\eta$ is the channel parameter, $d$ is the dimension of the underlying Hilbert space. This channel is $n$\textit{-incompatibility breaking} for all $ 0 \le \eta \le \frac{n+d}{n(d+1)}$. With $d=2$, $\mathcal{W}_{\eta}$ is $\tibc$ and $\thibc$ for $\eta\le 0.66$ and $0.55$, respectively \cite{heinosaari2015incompatibility}.\bigskip

\subsubsection{Incompatibility of generalized spin-observables}

 Consider a spin observable $A = \hat{a} \cdot \vec{\sigma}$ with projectors $P_{\pm} (\vec{a}) = \frac{1}{2}(\mathbf{I} \pm A)$, with $|\vec{a}|=1$. In  presence of a noise channel $\mathcal{E}$, these projectors are mapped to the \textit{noise induced} POVM $P_\pm (\vec{\alpha})$,  by the transformation 

\begin{equation}\label{eq:observable}
	A \xrightarrow[  ]{\mathcal{E}^*} A^{(x, \eta)}= x \mathbf{I}  + \eta \hat{a} \cdot \vec{\sigma}.
\end{equation}
Here, $x$ and $\eta$ characterize the bias and sharpness such that un-biased projective measurements correspond to $x=0$ and $\eta = 1$. Two biased and unsharp observables $A^{(x, \eta)}$ and $B^{(y, \xi)}$ are jointly measurable if \cite{Yu2010, Heinosaari2020}
\begin{equation}\label{eq:IBxeta}
	\left(1- S(x, \eta)^2 - S(y, \xi)^2 \right) \left( 1 - \frac{\eta^2 }{S(x, \eta)^2} - \frac{\xi^2 }{S(y, \xi)^2} \right) \le \left( \vec{\eta} \cdot \vec{\xi} - xy\right),
\end{equation}
where $\vec{\eta} = \eta \hat{a}$ and $\vec{\xi} = \xi \hat{b}$ are the unsharpness parameters, with 
\begin{equation}
	S(p,q) = \frac{1}{2} \left( \sqrt{\left(1+p\right)^2 - q^2} + \sqrt{\left(1-p\right)^2 - q^2} \right).
\end{equation}
For unbiased observables $A^{(0,\eta)}$ and $B^{(0,\xi)}$, the necessary and sufficient condition for compatibility simplifies to the  following
\begin{align}\label{eq:coplanar}
	|\vec{\eta} + \vec{\xi}| + |\vec{\eta} - \vec{\xi}|  \le 2.
\end{align}

When both the observables are subjected to identical noise channel i.e., $\xi  = \eta$  and $\vec{\eta}$ and $\vec{\xi}$ are perpendicular, then the pairwise joint measurability condition Eq. (\ref{eq:coplanar}) becomes
\begin{equation}\label{pjm}
	\eta\leq \frac{1}{\sqrt{2}}.
\end{equation}
This provides the condition for the incompatibility breaking of two observables and the corresponding (unital) channel is called an incompatibility breaking channel ($\tibc$). 
\subsection{Nonlocality}
\subsubsection{Bipartite nonlocality: CHSH inequality}
Consider the scenario of two spatially separated qubits with observables  $\hat{A}_i=\sum_{k=1}^{3}{\hat{a_{ik}}} \cdot {\hat{\sigma_k}}$  and $\hat{B}_j=\sum_{l=1}^{3}{\hat{b_{jl}}} \cdot {\hat{\sigma_l}}$ acting on each qubit respectively, where $\hat{a_{ik}}$ and  $\hat{b_{jl}}$ are unit vectors in $\mathbb{R}^3$; $i,j=1,2$; and the $\hat{\sigma}_i$'s are spin projection operators. The Bell operator associated with the CHSH inequality has the form
\begin{equation}
\label{Bell-op}
\hat{\Bb} = \hat{A}_1\otimes(\hat{B}_1 + \hat{B}_2) + \hat{A}_2\otimes(\hat{B}_1-\hat{B}_2), 
\end{equation}
Then the Bell-CHSH inequality, for any state $\rho $ is 
\begin{equation}
\label{bell-ineq}
\Trace\left[\rho ~\hat{\Bb}(\hat{A}_{1},~\hat{A}_{2},~\hat{ B}_{1},~\hat{B}_{2}) \right] \leq 2.
\end{equation}
where the observables on Alice's and Bob's  side are pairwise compatible. The violation of the above inequality (\ref{bell-ineq}) is sufficient to justify the nonlocality of the  quantum state.   Since  incompatible observables acting on entangled particles enable nonlocality, thus incompatibility is necessary for violation of (\ref{bell-ineq}). \bigskip

The essence of  Eq. (\ref{eq:SHpictures}) is that  the effect of noise as decoherence of the nonlocal resource $\rho$, can be interpreted as distortion of Alice's local measurement resource (incompatibility).  This  fact was exploited by Pal and Ghosh \cite{pal2015non} to introduce the notion of   CHSH-NBC defined below.

\begin{definition}
	Any qubit channel $\mathcal{E}:\mathcal{E}({\mathcal{H}}^i) \rightarrow {\mathcal{E}}({\mathcal{H}}^i)$ is  said to be  $\nbc$ if applying on one side of (arbitrary) bipartite  state $\rho_{AB}$, it produces a state $\rho'_{AB}= (\id \otimes \mathcal{E}) (\rho_{AB})$ which satisfies the Bell-CHSH inequality (\ref{bell-ineq}).
\end{definition}
This means that for any choice of POVMs  $\{\pi_{a|x}^{A}\}$ and 
	$\{\pi_{b|y}^{B}\}$ on subsystems $A$ and $B$, respectively,  there exist conditional distributions $P(a|x,\lambda)$ and $P(b|y,\lambda)$ and shared variable $\lambda$, such that
\begin{eqnarray}
\Tr\left[( \pi_{a|x}^{A}\otimes\pi_{b|y}^{B} )\sigma^{AB} \right]=\int d\lambda ~p(\lambda) ~P(a|x,\lambda) ~P(b|y,\lambda).
\end{eqnarray}
The violation of the Bell inequalities in the measurement statistics $P(ab|xy,\lambda)=\Tr[( \pi_{a|x}^{A}\otimes\pi_{b|y}^{B} )\sigma^{AB}]$ indicates that $\sigma^{AB}$ is not a local state. A unital channel is particularly important as it breaks the nonlocality for any state when it does so for maximally entangled states \cite{pal2015non}.\bigskip

\subsubsection{Tripartite nonlocality: Svelitchny inequality}

In a tripartite Bell scenario, with Alice, Bob, and Charlie performing measurements $A$, $B$ and $C$ having outcome $a$, $b$ and $c$  respectively, if the joint correlations can be written as
\begin{align}\label{eq:local}
P(abc|ABC) &= \sum_{l} c_l P_{l}(a|A)P_{l}(b|B) P_{l}(c|C)
\end{align}
with $0 \le c_l \le 1$ and $\sum_l c_l =1$, then they are local \cite{Mermin1990,Brunner2014bell}.
Svetlichny \cite{Svetlichny1987} proposed the \textit{hybrid} local nonlocal form of probability correlations
\begin{align}\label{eq:hybrid}
P(abc|ABC) &= \sum_{l} c_l P_{l}(ab|AB) P_{l}(c|C)  + \sum_{m} c_l P_{m}(ac|AC) P_{m}(b|B) \nonumber \\& + \sum_{n} c_l P_{n}(bc|C) P_{n}(a|A),
\end{align}
with $0 \le c_l, c_m, c_n \le 1$ and $\sum_l c_l + \sum_m c_m + \sum_n c_n =1$.
The quantum version involves triplets of particles subjected to independent dichotomic measurements with operators $\hat{A}_1$ and $\hat{A}_2$ for Alice, $\hat{B}_1$ and $\hat{B}_2$ for Bob, and $\hat{C}_1$ and $\hat{C}_2$ for Charlie, with  each measurement resulting in outcome $\pm 1$.  One defines the Mermin operator \cite{Mermin1990} as 
\begin{equation}
	\hat{M}  = \hat{A}_1 \otimes \hat{B}_1  \otimes \hat{C}_2 + \hat{A}_1 \otimes \hat{B}_2 \otimes \hat{C}_1 + \hat{A}_2 \otimes \hat{B}_1 \otimes \hat{C}_1 - \hat{A}_2 \otimes  \hat{B}_2 \otimes \hat{C}_2.
\end{equation}
Similarly, the Svetlichny operator is defined as \cite{Svetlichny1987,Cereceda2002}
\begin{eqnarray}\label{eq:SvetDef}
	\hat{S} &=&\hat{A}_1\otimes\big[(\hat{B}_1 + \hat{B}_2)\otimes \hat{C}_1 +(\hat{B}_1 - \hat{B}_2)\otimes \hat{C}_2\big]\nonumber\\
	&+& \hat{A}_2 \otimes\big[(\hat{B}_1 - \hat{B}_2)\otimes \hat{C}_1 -(\hat{B}_1 + \hat{B}_2)\otimes \hat{C}_2\big].
\end{eqnarray}
The respective average values of these operators are classically upper bounded in the form of the Mermin and Svetlichny inequalities written as
\begin{equation}
|\langle \hat{M} \rangle | \le 2, \qquad |\langle \hat{S} \rangle | \le 4.
\end{equation}  
The maximum quantum bounds are  known to be $2\sqrt{2}$ and $4\sqrt{2}$ for $|\langle \hat{M} \rangle |$ and $|\langle \hat{S} \rangle |$, respectively, and can be attained by, for example, the  GHZ state \cite{Mitchell2004}. Several approaches have been adopted to find the maximum quantum value of the Svetlichny operator \cite{SGhose2009,Su2018}. Recently, \cite{Siddiqui2019,li2017tight} have analytically found the tight upper bound of the Mermin and Svetlichny operator as given below.

\begin{definition}\label{def:SIQmax}
	For any three-qubit quantum state $\rho$, the maximum quantum value of the Mermin and  the Svetlichny operator is bounded as {\rm \cite{li2017tight,Siddiqui2019} }
	\begin{eqnarray}
		\label{sveq}
	\max |\langle \hat{M} \rangle_\rho| \leq 2\sqrt{2} \lambda_1, \qquad	\max |\langle \hat{S} \rangle_\rho| \leq 4\lambda_1,
	\end{eqnarray}
	where $\langle\hat{M }\rangle_\rho=\Tr[\hat{M}\rho]$,  $\langle\hat{ S }\rangle_\rho=\Trace[\hat{S}\rho]$ and $\lambda_1$ is the maximum singular value of the matrix $M=(M_{j,ik})$, with  $M=(M_{ijk}=\Trace[\rho(\sigma_i\otimes\sigma_j\otimes\sigma_k)]$, $i,j,k=1,2,3$.
\end{definition}

 We refer the reader to \cite{Siddiqui2019,li2017tight} for the class of states that saturate the above inequalities.  We  make use of above mentioned bounds to study the nonlocality breaking property of (unital) channels acting on one (or more) party  in the tripartite scenario. 

\section{Results and discussion}\label{sec:results}
In what follows, we will make use of the fact that every qubit channel $\mathcal{E}$ can be represented in the Pauli basis $\{ \sigma_0, \sigma_1, \sigma_2, \sigma_3 \}$, where $\sigma_0 = \mathbf{I}$, by a unique $4\times 4$ matrix $\mathbb{M}_{\mathcal{E}}  =[1, \bm{0}; \textbf{t}, \textbf{T}]$ \cite{Ruskai2002,jn2021}. Here $\textbf{T} = {\rm diag.}[\eta_1, \eta_2, \eta_3]$ is a real diagonal matrix and $\bm{0} = (0~0 ~0)$ and $\textbf{t} = (t_1~t_2~t_2)^T$ are row and column vectors, respectively. For $\mathcal{E}$ to be unital i.e., $\mathcal{E}(\mathbf{I}) = \mathbf{I}$, we must have $\textbf{t} = (0~0 ~0)^T$. The conjugate map $\mathcal{E}^\dagger$ is characterized by $\mathbb{M}^\dagger_{\mathcal{E}}  = [1, \textbf{t}^T; \bm{0}^T , \textbf{T}]$, such that the action on a state $\rho = \frac{1}{2}(\mathbf{I} + \textbf{w} \cdot \bm{\sigma})$   is  given by 
\begin{align}
	&\mathbb{M}_{\mathcal{E}} : \mathbf{I} \rightarrow \mathbf{I} + \textbf{t}\cdot \bm{\sigma}, \quad \sigma_j \rightarrow \eta_j \sigma_j, \label{eq:Map} \\
	&\mathbb{M}^\dagger_{\mathcal{E}} : \mathbf{I} \rightarrow \mathbf{I}, \quad \sigma_j \rightarrow  t_j \mathbf{I} + \eta_j \sigma_j. \label{eq:ConjMap}
\end{align}
\subsection{Equivalance of CHSH nonlocality breaking and incompatibility breaking channels}
Our first result establishes an \textit{equivalence} of the CHSH nonlocality breaking channel acting on one party, with  its \textit{dual} being an incompatibility breaking channel -- in the context of $\tibc$s. The result can be summarized by the following two theorems: 
\begin{theorem}
	If the conjugate of a qubit channel ${\rm \mathcal{E}}$  is $\tibc$ , then the channel itself is CHSH-NBC.
	\end{theorem}
	\begin{proof}
		Consider the Bell-CHSH inequality given in (\ref{bell-ineq}), such that  $[A_{1}, A_{2}] \ne 0$ and   $[B_{1}, B_{2}] \ne 0$, i.e., the operators $A_1$ and  $A_2$ and $B_1$ and  $B_2$ are incompatible in \textit{conjunction}.   Let   $\mathcal{E}^\dagger $ be the conjugate channel that is $\tibc$.  Then the action of this channel on Alice's side makes $A_1$ and $A_2$ compatible, i.e., $[A_{1}, A_{2}]=0$.  Therefore, the Bell-CHSH inequality is not violated \cite{wolf2009measurements}, and we have
		\begin{equation}
 	  \Trace\left[\rho \Bb\left(\mathcal{E}^\dagger[A_1],\mathcal{E}^\dagger[A_2], B_{1},B_{2}\right) \right]\leq 2.
 	  \end{equation}
   Alternatively, 
   \begin{align}
   &	\Trace\left[\rho \mathcal{E}^\dagger[A_1] \otimes  B_{1} \right]+\Trace \left[\rho \mathcal{E}^\dagger[A_2] \otimes B_{1} \right] \nonumber  \\& +\Trace \left[\rho \mathcal{E}^\dagger[A_1] \otimes B_{2} \right]-\Trace \left[\rho \mathcal{E}^\dagger[A_2] \otimes  B_{2} \right]\leq 2.
   \end{align}
   This can be viewed in the Schr{\"o}dinger picture as 
   \begin{align}
		     &\Trace\left[(\mathcal{E} \otimes \id)[\rho]  A_1 \otimes B_{1} \right]+\Trace\left[(\mathcal{E} \otimes \id)[\rho]  A_2 \otimes B_{1} \right] \nonumber \\&  + \Trace\left[(\mathcal{E} \otimes \id)[\rho]  A_1 \otimes B_{2} \right]-\Trace\left[(\mathcal{E} \otimes \id)[\rho]  A_2 \otimes B_{2} \right]\leq 2,
		\end{align}
	  which tells us that the CHSH inequality is satisfied even when  operators $A_1$ and $A_2$ and $B_1$ and $B_2$ are incompatible in \textit{conjunction}. Therefore, the  action of $\mathcal{E}$ (to be precise of $\mathcal{E} \otimes \id$) on state $\rho$ is solely responsible for nonviolation of the CHSH inequality. We conclude that $\mathcal{E}$ is CHSH-NBC.
		\end{proof}
\begin{theorem}
 If a qubit channel $\mathcal{E}$  is CHSH-NBC, then its conjugate is  $\tibc$.
\end{theorem}

	\begin{proof}
	Here we start with incompatible operators associated with the respective subsystems, $[A_1,A_2] \ne 0$ and $[B_1,B_2] \ne 0$ and assume that the channel $\mathcal{E}$ acting on Alice's side does not allow for the violation of the  CHSH inequality, that is
	\begin{equation}
		\Trace\left[\left( \mathcal{E} \otimes \id \right) [\rho] \Bb\left( A_{1},A_{2}, B_{1}, B_{2} \right) \right]\leq 2.
	\end{equation}
 In other word, looking from the measurement point of view, in the Heisenberg picture,  we have
	\begin{align}
 &\Trace\left[\rho \mathcal{E}^\dagger [A_{1}]\otimes B_1 \right]+\Trace\left[\rho \mathcal{E}^\dagger[A_{1}]\otimes B_2 \right]\nonumber \\&+\Trace\left[\rho \mathcal{E}^\dagger [A_{2}]\otimes B_1 \right]-\Trace[\rho \mathcal{E}^\dagger [A_{2}] \otimes B_2]\leq 2
	\end{align}
 The above inequality holds for arbitrary  state $\rho$ which can even be an entangled state. Thus the nonviolation of the CHSH inequality is coming from the action of $\mathcal{E}^\dagger $ on the operators $A_1$ and $A_2$, making them compatible,  $\left[ \mathcal{E}^\dagger [A_1], \mathcal{E}^\dagger [A_2] \right] =0$. We conclude that  $(\mathcal{E}^\dagger \otimes \id)$ is incompatibility breaking. 
	\end{proof}

In order to verify the above  results, let us  consider the CHSH inequality (\ref{bell-ineq}), such that the local observables are subject to some biased noise characterized by $(x_a, \eta_a)$ and $(x_b, \eta_b)$. As a result, the observables are modified as  $A_k =  \mathbf{I} x_a + \eta_a  \hat{a}_k \cdot \vec{ \sigma}_k$,  $B_k =  \mathbf{I} x_b + \eta_b  \hat{b}_k \cdot \vec{\sigma}_k$, and  $k=1,2$.  Averaging with respect to the singlet-state, the  CHSH inequality is satisfied if
\begin{align}
	2 x_a x_b+ \eta_a \eta_b \left| \cos\theta_{11} + \cos\theta_{12} + \cos\theta_{21} - \cos \theta_{22} \right| \le 2.
\end{align}
  The modulus term has a maximum value of  $2 \sqrt{2}$ for $\theta_{ab} = \theta_{a^\prime b} = \theta_{ab^\prime} = \theta$ and $\theta_{a^\prime b^\prime} = 3 \theta$ with $\theta = \pi/4$. The above inequality becomes \cite{KAR199512}
\begin{equation}
	\eta_a \eta_b  \le \frac{1 - x_a x_b}{\sqrt{2}}.
\end{equation}
For the unbiased noise on  \textit{either side},  i.e., $x_a=0$ or $x_b=0$, the incompatibility breaking condition   reduces to 
\begin{equation}
   \eta_a \eta_b  \le \frac{1}{\sqrt{2}}.
\end{equation}
Further, if this unbiased noise is acting only on one side (for example, on $A$) of the bipartite system, then $\eta_b = 1$, and we have 
\begin{equation}\label{eq:CHSHnlb}
	\eta_a   \le \frac{1}{\sqrt{2}}.
\end{equation}
 Having seen the impact of noise on the violation of the CHSH inequality, we now look at the incompatibility breaking condition for observables $A_1$ and  $A_2$ subjected to the same noise. As a result we have $A_1^{(x_a, \eta_a)}$ and $A_2^{(x_a, \eta_a)}$ as biased and unsharp observables. Using inequality (\ref{eq:IBxeta}), the incompatibility condition for these observables reads 
\begin{equation}
	\eta_a  \le \frac{1 - x_a^2}{\sqrt{2}}.
\end{equation}
This  coincides with the NB condition (\ref{eq:CHSHnlb}) only if $x_a = 0$. Since the action of  the nonunital channel on a projector results in a  biased observable \cite{jn2021}, one finds that the conditions for nonlocality breaking  and incompatibility breaking agree as long as the dynamics is unital in the  Schr{\"o}dinger picture. \bigskip 

\subsection{Nonlocality and incompatibility breaking channels in tripartite scenario}	
In the tripartite scenario, the  aforesaid relation between nonlocality and incompatibility does not hold in general. In this section, we therefore first obtain the NB condition for some well known tripartite states and then identify the range of channel parameter where NB agrees with the $2$-IBC condition. Let us first introduce the definition of the absolute and genuine tripartite NBC  as
\begin{definition}
	For any three-qubit state $\rho_{ABC}$, a given qubit channel $\mathcal{E}$ is said to be  absolute or  genuine \textit{nonlocality breaking} if acting on any qubit, it gives a state (for example) $\rho'_{ABC}=(\mathcal{E}\otimes I \otimes I)(\rho_{ABC})$, which satisfies Mermin inequality $\langle \hat{M} \rangle_{\rho'_{ABC}}\leq 2$ or  Svetlichny inequality $\langle \hat{S} \rangle_{\rho'_{ABC}}\leq 4$.
\end{definition}

 In terms of  the largest singular value in Definition \ref{def:SIQmax}, the Mermin nonlocality breaking condition (M-NBC) and Svetlichny nonlocality breaking condition (S-NBC) are, respectively,  given as
\begin{equation}
	\eta_M  \le \frac{1}{\sqrt{2} \lambda_{max} }, \qquad  \eta_S  \le  \frac{1}{\lambda_{max}}.
\end{equation}

	We now consider some well-known tripartite quantum states in which one party is subjected to a noisy evolution, and use Def. \ref{def:SIQmax} to obtain the conditions on the noise parameter for breaking the Mermin and Svetlichny nonlocality.\bigskip
	
\textbf{Example 1.} Our first example is the generalized GHZ state \cite{Dur2000,GHZ2007}
\begin{equation}\label{eq:GHZ}
|\psi_{GHZ} \rangle = \alpha | 000 \rangle + \beta |111 \rangle.
\end{equation}
The entanglement in this state is completely destroyed when any of the three qubits is traced out, i.e.,  $\Tr_k\left[ |\psi_{GHZ} \rangle \langle \psi_{GHZ}| \right] = \mathbf{I}/2$, with $k=A,B,C$.  Let us assume a unital noise channel  $\mathcal{E}$ acting on one qubit (for example, the first one)  according to Eq. (\ref{eq:Map}). This can be achieved by expressing $|\psi_{GHZ} \rangle $ in the  Pauli basis and invoking $\sigma_i \rightarrow \eta \sigma_i$  ($i=x,y,z$, $0 \le \eta \le 1)$ at the first qubit
\begin{align}
&\left(\mathcal{E} \otimes \mathbf{I} \otimes \mathbf{I} \right) \left[|\psi_{GHZ} \rangle \langle \psi_{GHZ}  | \right] \nonumber \\& =
\frac{1}{8}
\Big[ \mathbf{I} \otimes \mathbf{I} \otimes \mathbf{I} + \mathbf{I} \otimes \sigma_3 \otimes \sigma_3 +\eta( \sigma_3 \otimes \mathbf{I} \otimes\sigma_3)\nonumber\\&
+ (\alpha^2 -\beta^2)(\eta\sigma_3 \otimes \mathbf{I} \otimes \mathbf{I} + \mathbf{I} \otimes \sigma_3 \otimes \mathbf{I}\nonumber\\&
+\mathbf{I} \otimes \mathbf{I} \otimes \sigma_3 +\eta \sigma_3 \otimes \sigma_3 \otimes \sigma_3) + 2\alpha\beta\eta (\sigma_1 \otimes \sigma_1\otimes \sigma_1\nonumber\\&
-\sigma_1 \otimes \sigma_2 \otimes \sigma_2 -\sigma_2 \otimes \sigma_1 \otimes \sigma_2 -\sigma_2 \otimes \sigma_2 \otimes \sigma_1) \Big]
\end{align}
The matrix  $M_{j, ik}$  in Def. \ref{def:SIQmax} turns out to be
\begin{align}
M_{j,ik} =
\small \begin{pmatrix}
2\eta \alpha\beta & 0  & 0 & 0 &  -2\eta \alpha\beta & 0 & 0 & 0 & 0\\
0 & -2\eta \alpha\beta  & 0 & -2\eta \alpha\beta & 0  & 0 & 0 & 0 & 0\\
0 & 0 & 0 & 0 & 0 & 0  & 0 & 0 & \eta(\alpha^2-\beta^2)
\end{pmatrix}.\nonumber\\ 
\end{align}
\normalsize
This has singular values $2\sqrt{2}\eta\alpha\beta$, $2\sqrt{2}\eta\alpha\beta$, and $\eta \left| \alpha^2 - \beta^2 \right|$.
Thus the condition for the channel to be Mermin and Svetlichny NB is given by,
\begin{eqnarray}\label{eq:ghzNB1}
\eta_M  \leq \frac{1}{4\alpha\beta}, \qquad \eta_S  \leq \frac{1}{2\sqrt{2}\alpha\beta}.
\end{eqnarray} 	

\textbf{Example 2.} 
Next, we consider the well-known generalized W-state given by
\begin{equation}\label{eq:Wstate}
	|\psi\rangle=\alpha |100\rangle + \beta |010\rangle + \gamma |001\rangle,
\end{equation}
with $\alpha, \beta$, and $\gamma$ real and $\alpha^2 + \beta^2 + \gamma^2 =1$. This state is special in the sense that if one qubit is lost, the state of the  remaining two qubits is still entangled, unlike the GHZ state. The matrix  $M_{j, ik}$  in Def.\ref{def:SIQmax} corresponding to unital noise acting on the first qubit of the state is given by 
\begin{align}
	M_{j,ik}= 
	\begin{pmatrix}
		0 & 0  & \eta \omega & 0 & 0 & 0  & \eta \omega & 0 & 0\\
		0 & 0 & 0 & 0 & 0 &  \eta \omega &0 & \eta \omega & 0\\
		\eta \omega^\prime  & 0 & 0 & 0 &\eta \omega^\prime  & 0  & 0 & 0 & \eta(2\alpha\gamma-\beta^2)
	\end{pmatrix},\nonumber
\end{align} 
with $\omega = \alpha\beta+\beta\gamma$, $\omega^\prime = (\alpha^2+\gamma^2)$. 
 The largest singular value turns out to be $\lambda = \eta \sqrt{1 + 8 \beta^2 \gamma^2 }$. Now $\beta^2 + \gamma^2 = 1- \alpha^2 = k$ (for example), so that $\beta^2 \gamma^2 = \beta^2 (k - \beta^2)$. This quantity attains the maximum $k^2/4$ at $\beta^2 = k/2$,  leading to  $\lambda_{max} = \eta \sqrt{1 + 2 k^2 }$. The Mermin and Svetlichny nonlocality breaking conditions then read
\begin{equation}\label{eq:etaW}
\eta_M \le \frac{1}{\sqrt{2} \sqrt{1 + 2k^2}}, \qquad	\eta_S \le \frac{1}{\sqrt{1 + 2k^2}}.
\end{equation}

The singular values attain a maximum at $k=2/3$, i.e., $\alpha = \beta = \gamma = 1/\sqrt{3}$. Under this condition and  for $\eta_S$ approximately in the range $ [0.707- 0.727] $, the channel is S-NBC but not $\tibc$. Note that in the bipartite case, a unital channel is CHSH-NBC for all $\eta\leq 1/\sqrt{2}$  ~ \cite{pal2015non}. Thus in the range $\eta\approx 0.707-0.727 $, it is S-NBC but not CHSH-NBC.\bigskip
                                          
\textbf{Example 3.} We next consider the three-qubit partially entangled set of maximal slice (MS) states \cite{Carteret2000,liu2016controlled}
\begin{equation}\label{eq:Slice}
	|\psi_{MS}\rangle = \frac{1}{\sqrt{2}} \bigg[|000\rangle+|11 \bigg(\alpha|0\rangle+ \beta|1\rangle \bigg) \bigg],
\end{equation}
 where $\alpha$ and $\beta$ are real with $\alpha^2 + \beta^2= 1$.  The matrix $M_{j,ik}$ corresponding to Eq. (\ref{eq:Slice}) after the action of unital noise on the first qubit is given by
 	\begin{align}
 	&M_{j,ik}&=
 	\begin{pmatrix}
 	\eta\beta & 0  & \eta\alpha & 0 &-\eta\beta &0  & 0 & 0 & 0\\
 	0 & -\eta\beta & 0 & -\eta\beta & 0 &  -\eta\alpha &0 & 0 & 0\\
 	0 & 0 & 0 & 0 & 0 & 0  & \eta\alpha\beta & 0 & \eta(\frac{1+\alpha^2-\beta^2}{2})
 	\end{pmatrix}\nonumber
 	\end{align} 
 with three singular values $\lambda_1=\eta\sqrt{\frac{1}{4} \left(\alpha ^2-\beta ^2+1\right)^2+ (\alpha  \beta )^2}$ and two equal singular values $\lambda_2= \lambda_3=\eta\sqrt{(\alpha^2+2\beta^2)}$ respectively, leading to the condition for Mermin and Svetlichny nonlocality breaking as
\begin{eqnarray}\label{eq:Slice1}
\eta_M \leq \frac{1}{\sqrt{2(\alpha^2+2\beta^2)}}, \qquad \eta_S \leq \frac{1}{\sqrt{(\alpha^2+2\beta^2)}}.
\end{eqnarray}
\bigskip

\begin{figure}
	\centering
	\includegraphics[width=70mm]{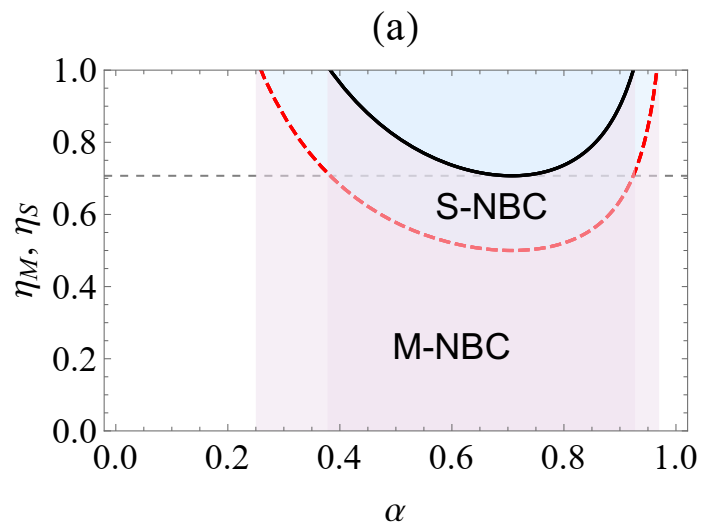} \\
	\includegraphics[width=70mm]{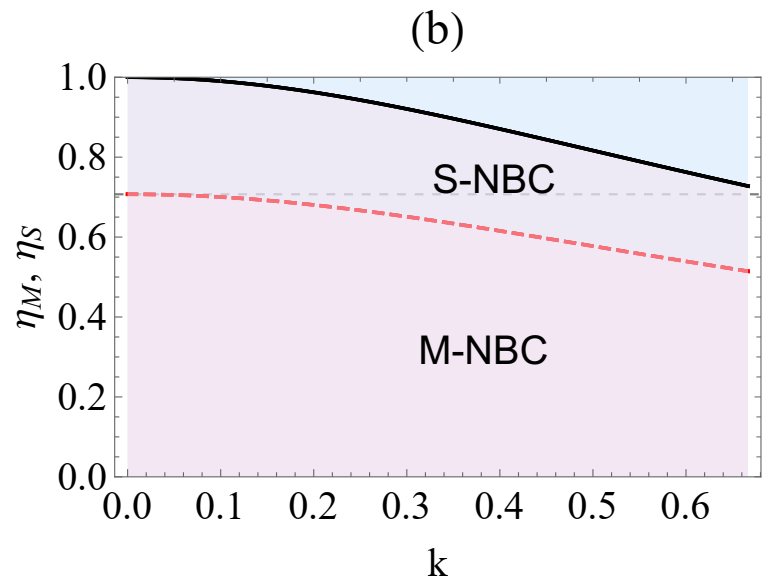}\\
	\includegraphics[width=70mm]{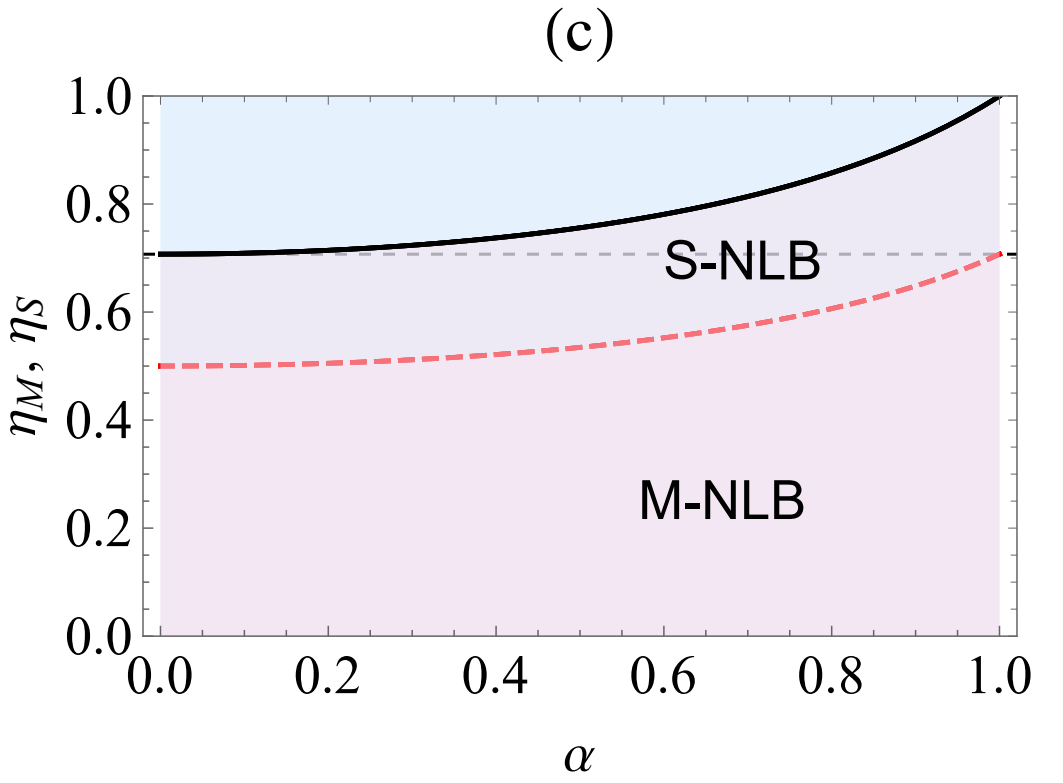}\\
	\includegraphics[width=70mm]{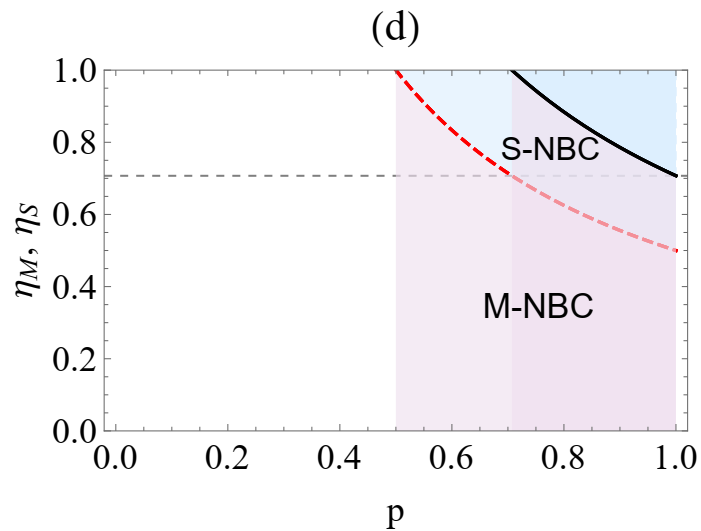}
	\caption{(Color online): The region below the dashed (red) and  solid (black)  curve in (a), (b), (c), and (d) corresponds to M-NBC and S-NBC given by Ineqs. (\ref{eq:ghzNB1}), (\ref{eq:etaW}), (\ref{eq:Slice1}), and (\ref{eq:etaGHZmix}), respectively, plotted against the (dimensionless) state coefficients.  The pairwise incompatibility breaking condition (\ref{pjm}) pertains to all points below the horizontal dashed line.}
	\label{fig:eta}
\end{figure}	
\bigskip

\textbf{Example 4.} Consider the quantum state \cite{Augusiak2015} 
\begin{equation}\label{eq:GHZmixed}
\sigma_{GHZ} =p|GHZ\rangle\langle GHZ|+(1-p) I_2 \otimes \tilde{I}.
\end{equation}
Here, $ |GHZ\rangle\langle = (|000\rangle + |111\rangle)/\sqrt{2}$,   $\tilde{I} = diag.(1, 0, 0, 1)$, is diagonal matrix, and  $0\leq p \leq 1$. In fact,  $\tilde{I} = |\Phi^{+} \rangle \langle \Phi^{+}| + |\Phi^{-} \rangle \langle \Phi^{-}|$, where $ \Phi^{\pm}$ are Bell states, is a separable state, and it tells us that $\sigma_{GHZ}$ has with probability $p$ the GHZ state and with probability $(1-p)$ the first qubit is left in the mixed state and the second and third qubit are in separable state  $\tilde{I}$.  The corresponding  $M_{j,ik}$ matrix with a unital noise acting on one qubit is given by 
\begin{eqnarray}
M_{j,ik} =
\begin{pmatrix}
p\eta & 0  & 0 & 0 & -p\eta & 0 & 0 & 0 & 0\\
0 & -p\eta & 0 & -p\eta & 0 & 0 & 0 & 0 & 0\\
0 & 0 & 0 & 0 & 0 & 0  & 0 & 0 & 0
\end{pmatrix},\nonumber
\end{eqnarray}
with  singular values $\lambda_1=\lambda_2=\sqrt{2}\eta p$. Thus the condition for the channel to be Mermin and Svetlichny NBC becomes 
\begin{eqnarray}
\label{eq:etaGHZmix}
\eta_M \leq \frac{1}{2 p}, \qquad \eta_S \leq \frac{1}{\sqrt{2}p}.
\end{eqnarray}
The conditions for Mermin and Svetlichny nonlocality breaking channel given by Ineqs. (\ref{eq:ghzNB1}), (\ref{eq:Wstate}), (\ref{eq:Slice1}), and  (\ref{eq:etaGHZmix}) obtained by the application of a unital quantum channel to one party of a tripartite system,  are depicted in Fig. \ref{fig:eta}. All the points below the solid (black) and dashed (red) curve correspond to the nonlocality breaking channel, while as the points below  the horizontal dashed line, $\eta=1/\sqrt{2}$, pertain to pairwise incompatibility breaking. In all the four examples, Fig.~ \ref{fig:eta} (a)-(d), the minimum value of $\eta_S$ for which SI is violated  is  $1/\sqrt{2}$, suggesting that \emph{genuine} nonlocal correlations can not be established if at least one pair of observable is compatible. The converse is not true, since there exist regions (above the horizontal dashed line and below the solid (black) curve) of Svetlichny nonlocality breaking even when the channel is not $\tibc$. Thus, these examples illustrate that \textit{corresponding to $2$-IBC the conjugate channels  are definitely S-NLB; however, the conjugate of S-NB channels may not necessarily be a $2$-IBC}. However, in the context of Mermin nonlocality, even the first statement does not hold, that is, \textit{existence of a $2$-IBC does not necessarily guarantee a conjugate channel that is M-NB.} Also, the minimum $\eta_M$ (that is maximum noise) for which a channel is M-NB  is always less by a factor of $1/\sqrt{2}$ than the minimum noise below which that channel is S-NB.
 It is worth pointing  here to  Fig.~ \ref{fig:eta} (c),  which  illustrates  that if a (unital) channel breaks the Svetlichny nonlocality for the GHZ state (which in fact violates the SI maximally) then it also does so for the mixture (\ref{eq:GHZmixed})  for $1/\sqrt{2}< p < 1$. This is unlike the bipartite scenario where a unital channel that breaks the CHSH nonlocality for the maximally entangled states is guaranteed to do so for all other states \cite{pal2015non}.  Summarizing, it is clear from the above examples that the existence of M-NBC or S-NBC does not guarantee the existence of a conjugate $2$-IBC, unlike the CHSH scenario. In particular, with the (unital) noise acting on one party of the $W$-state, $|\psi\rangle=(|100\rangle + |010\rangle + |001\rangle$)/$\sqrt{3}$, there exists a range of the channel parameter $\eta \in (1/\sqrt{2},3/\sqrt{17})$, where the channel is S-NBC but not CHSH NBC as depicted in the Fig. ~\ref{figsinbc}. Thus the Bell-CHSH inequality seems to be more suitable for a study of the incompatibility of observables than the multi-partite Bell-type inequalities.\bigskip

 Note that instead of one party, if two or all the three parties are subjected to noise, the conditions (\ref{eq:ghzNB1}), (\ref{eq:etaW}),  (\ref{eq:Slice1}), and (\ref{eq:etaGHZmix}), become
 \begin{equation}\label{eq:etaN}
 		\eta_M \le 	   \left( \frac{1}{\sqrt{2}\lambda_{max}}\right)^{1/n}, \qquad \eta_S \le 	\left( \frac{1}{\lambda_{max}}\right)^{1/n},
 \end{equation}
where $n$ corresponds to the number of qubits subjected to noise.  Since $1/\lambda_{max} < (1/\lambda_{max})^{1/2} < (1/\lambda_{max})^{1/3}$ (with $\lambda_{max} > 1$), the solid (black) and dashed (red) curves in  Fig.  \ref{fig:eta} (a)-(d)  are shifted up, thereby decreasing the region of nonlocality with increase in $n$.\bigskip

 \textit{General three qubit state subjected to general unital noise:} 
The (unital) noise acting on a single party considered in the above analysis involving the tripartite system assumes identical effects on $\sigma_x$, $\sigma_y$, and $\sigma_z$ corresponding to that party, in the sense that $\sigma_k \rightarrow \eta \sigma_k$ for all $k = x, y, z$. A more general transformation would take the particular party's $\sigma_k \rightarrow \eta_k \sigma_k$,  with $\sqrt{\eta_x^2 + \eta_y^2 + \eta_z^2} = \eta$, and $0 \le \eta \le 1$, such that 
	\begin{equation}
		\Phi\left( \frac{ \id + \vec{r} \cdot \vec{\sigma}}{2}\right) = \frac{\id + (\textbf{T}\vec{r}) \cdot \vec{\sigma}}{2},
	\end{equation}
where $\textbf{T} = {\rm diag.} [\eta_x, \eta_y, \eta_z]$, is a real diagonal matrix. The map $\Phi$ is completely positive for \cite{king2001minimal}
	\begin{equation}\label{eq:CP}
		\left| \eta_x \pm \eta_y \right| \le \left| 1 \pm \eta_z \right|,
	\end{equation}
	    which is a set of four inequalities and defines a tetrahedron in $\eta_{x} -\eta_{y} -\eta_{z}$ space. Under such a transformation, the singular values for the GHZ in (\ref{eq:GHZ}) and $W$ state in (\ref{eq:Wstate}) are, respectively, given by 
\begin{align}
&  \left(  2\sqrt{2} \alpha \beta \eta_x, 2\sqrt{2} \alpha \beta \eta_y, \left| \alpha^2 - \beta^2 \right| \eta_z \right), \\ {\rm and}~~\nonumber \\
	& \left( 2 \alpha \sqrt{\beta^2 + \gamma^2} \eta_x, 2 \alpha \sqrt{\beta^2 + \gamma^2} \eta_y, \sqrt{1 + 8 \beta^2 \gamma^2} \eta_z\right).
\end{align}
with the three singular values depending linearly on the respective noise parameters. Depending on which singular value is the largest, one can draw similar conclusions about nonlocality and incompatibility breaking properties of the noise channel as  in the case with uniform noise action. However, for a general three-qubit input state, the dependence of the singular values on noise parameters tuns out to be complicated leading to different conclusions regarding the nonlocality breaking property of such a channel. Let us consider the general situation where one would like to make a statement about the limiting noise beyond which no Mermin or Svetlichny nonlocal correlations can be established \emph{irrespective} of the state chosen. In this direction, we make use of the canonical five term decomposition of the three-qubit state \cite{acin2001three}: 
\begin{equation}
	\ket{\psi} = \lambda_0 \ket{000} + \lambda_1 e^{i\phi}  \ket{100} + \lambda_2 \ket{101} + \lambda_3 \ket{110} + \lambda_4 \ket{111}, 
\end{equation}
with $\lambda_i$ and $ \phi $ real parameters and $\sum_i \left| \lambda_i \right|^{2}= 1$ and $0 < \phi < \pi$. In the Pauli basis, one may denote the density matrix corresponding to $\ket{\psi}$ as $|\psi\rangle \langle \psi | \left[ \sigma_i \otimes \sigma_j \otimes \sigma_k\right]$ with Pauli matrices $\sigma_i$ $i=x,y,z$.  If a unital noise acts on one party (for example, the first) of such a state, we would have
\begin{align}\label{eq:SVAcin}
|\psi\rangle \langle \psi | \left[\sigma_i \otimes \sigma_j \otimes \sigma_k \right] \rightarrow |\psi\rangle \langle \psi | \left[ \eta_i  \sigma_i \otimes \sigma_j \otimes \sigma_k \right].
\end{align}

\begin{figure}
	\centering
	\includegraphics[width=90mm]{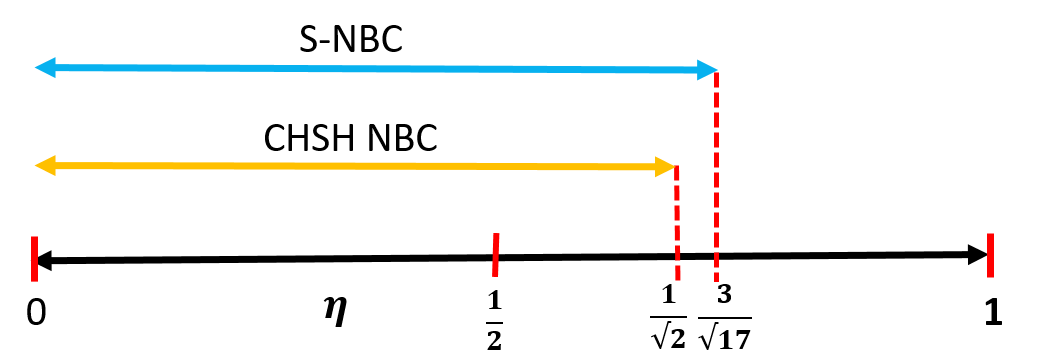} \\
	\caption{(Color online): Depicting the range of channel parameter $\eta \in  (\frac{1}{\sqrt{2}},\frac{3}{\sqrt{17}})$  in which it is S-NBC (for the $W$ state) but not CHSH-NBC.}
	\label{figsinbc}
\end{figure}

\begin{figure}
	\centering
	\includegraphics[width=82mm]{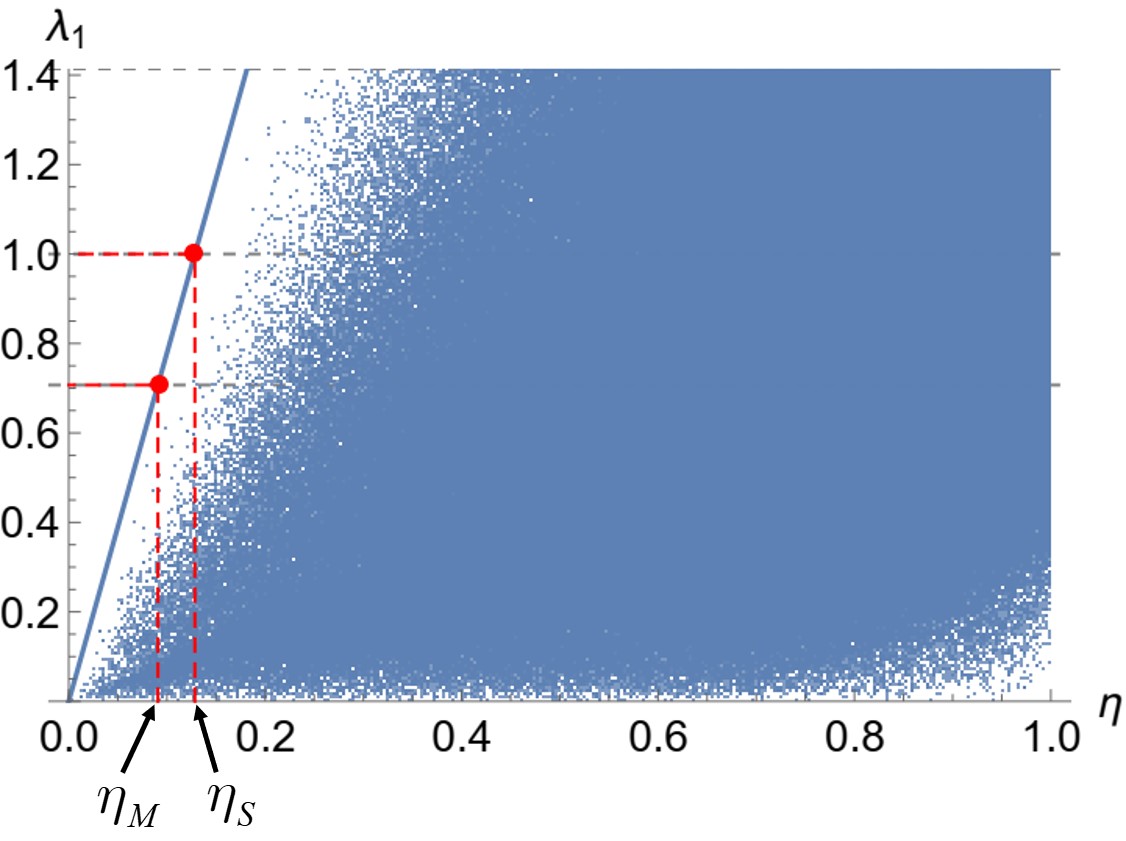}\\ \vspace{4mm}
	\includegraphics[width=82mm]{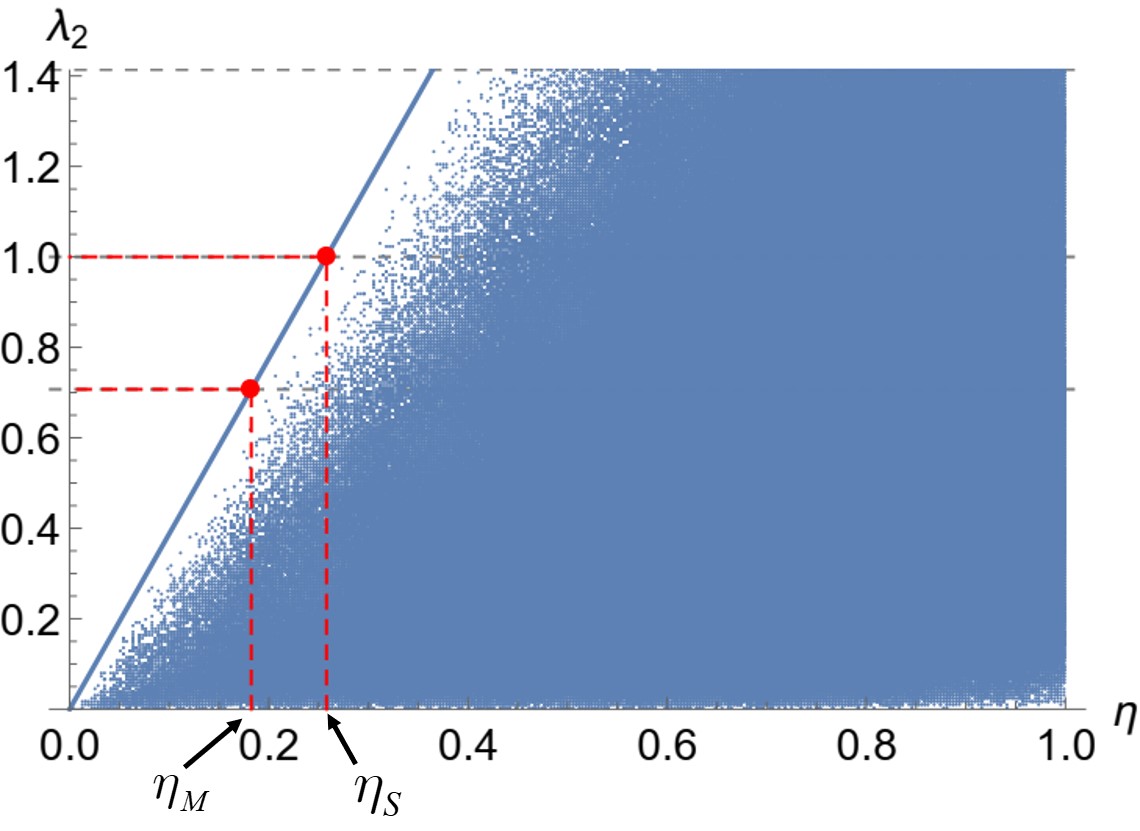}\\  \vspace{4mm}
	\includegraphics[width=82mm]{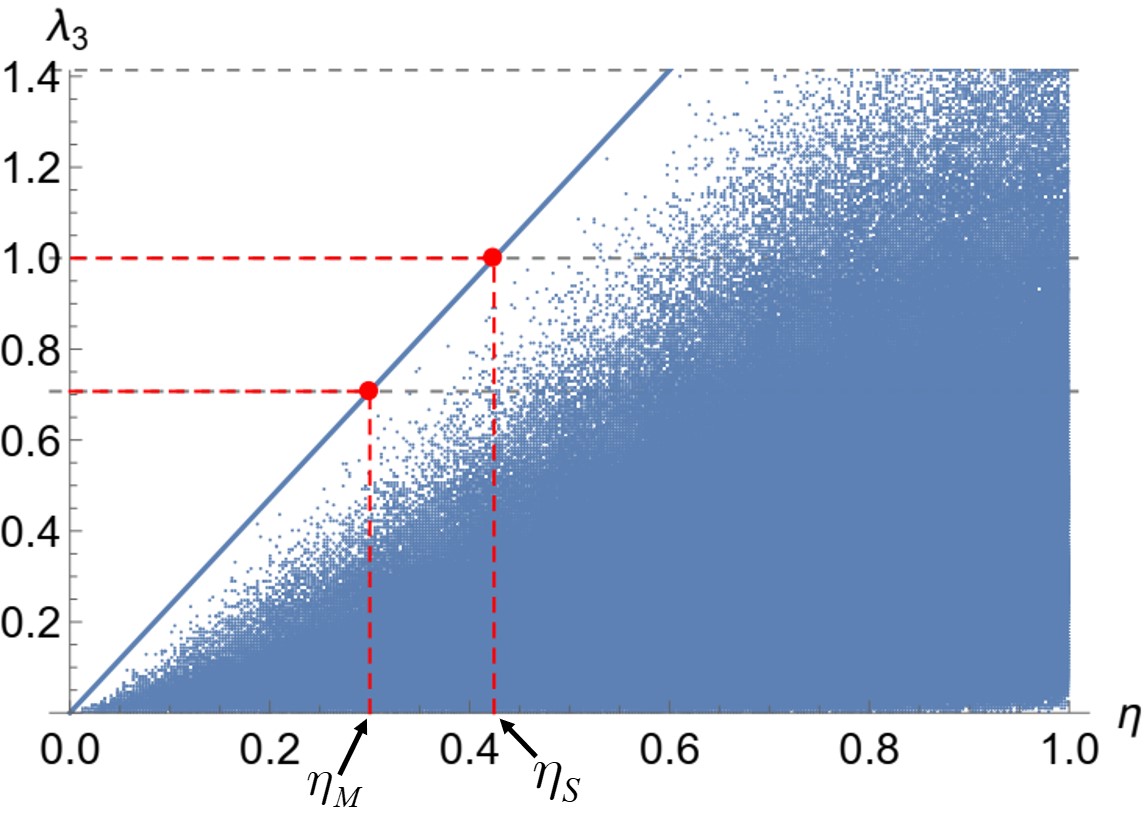}	
	\caption{(Color online) Singular values corresponding to a sample of $5\times 10^6$ randomly generated states of the form $(\ref{eq:SVAcin})$ plotted with respect to the dimensionless parameter $\eta = \sqrt{\eta_{x}^2 + \eta_{y}^2 + \eta_{z}^2}$.  In terms of these singular values, the conditions for  establishing the  Mermin and Svetlichny nonlocal correlations are given in Def. \ref{def:SIQmax} which  demands that the largest singular value be greater than $1/\sqrt{2}$ and $1$, respectively. The corresponding bound on the unsharpness parameters $\eta_M$ and $\eta_S$ trivially follow  and are summarized in the text. }
	\label{fig:SV}
\end{figure}

 We can now check the nonlocality breaking properties of such a general noise channel using Def. \ref{def:SIQmax} based on singular values of matrix  $M$.  For the above state, the singular values calculated according to Def. \ref{def:SIQmax}, are plotted in Fig. \ref{fig:SV} with respect to $\eta = \sqrt{\eta_x^2 + \eta_y^2 + \eta_z^2}$. The simulation makes use of  a sample of $5\times 10^6$ randomly generated states and the corresponding noise parameters $\eta_i$ subjected to  $0\le \sqrt{\eta_x^2 + \eta_y^2 + \eta_z^2} \le 1$ and also satisfying the completely positive condition  (\ref{eq:CP}).   According to the condition (\ref{sveq}), the Mermin  and Svetlichny nonlocal correlations are established only if the largest singular value is greater than $\frac{1}{\sqrt{2}}$ for the former and $1$ for the later. In Fig. \ref{fig:SV} the three singular values $\lambda_1, \lambda_2$, and $ \lambda_3$ are depicted with respect to parameter $\eta$. One finds that below a minimum $\eta_M$ and $ \eta_S$, the singular values do not exceed $\frac{1}{\sqrt{2}}$ and $1$, respectively. These values turn out to be $(\eta_M, \eta_S) = (0.090,0.128)$ for $\lambda_1$, $(\eta_M, \eta_S) = (0.182, 0.259)$ for $\lambda_2$, and $(\eta_M, \eta_S) = (0.300,0.409)$  for $\lambda_3$, and one may conclude that no Mermin (Svetlichny) nonlocal correlations are supported by  the unital noise channel if the noise parameter is below $0.090 ~(0.128)$, irrespective of the input state.

	\section{Conclusion}\label{sec:Concl}
	This paper is devoted to a study of the interplay between nonlocality breaking and incompatibility breaking power of noisy  quantum qubit channels. The action of quantum channels on projective measurements transforms them into noisy POVMs, characterized in particular by unsharpness parameters. As a consequence, noise tends to increase the compatibility of observers that are otherwise incompatible. In fact, pairwise incompatibility breaking is assured if the channel  parameter is less than or equal to $1/\sqrt{2}$. To be specific, we consider bipartite and tripartite scenarios, with CHSH nonlocality in the former and  Mermin and Svetlichny nonlocality in the later case. The degree of incompatibility breaking directly depends on the unsharpness parameters. Here, we showed that in the Bell-CHSH scenario, if the conjugate of a   channel is incompatibility breaking  then the channel is itself nonlocality breaking and the converse is also true. In the tripartite scenario, however, such an equivalence between nonlocality breaking and incompatibility breaking does not exist. We then consider various examples of three-qubit states and identify the state parameters for which the equivalence of nonlocality breaking corroborates with  the  pairwise incompatibility. In particular, it is illustrated that  the conjugate of incompatibility breaking channels is nonlocality breaking, however, the nonlocality breaking channels do not guarantee the existence of  conjugate channels that are incompatibility breaking.   This may be viewed as a useful feature of the Bell-CHSH inequality when it comes to the study incompatibility of observables.  Further, from  randomly generated three-qubit states subjected to general unital channels, we conclude that no Mermin (Svetlichny) nonlocal correlations are supported for $\eta_M <0.090$  ($\eta_S < 0.128$), $\eta_{M/S}$ being the channel unsharpness parameter.

	The channel activation of nonlocality in the CHSH scenario has 	been studied in  \cite{Zhang2020}. A future extension of this paper  could be the study of activation of Mermin and Svetlichny nonlocality under more general noise models and with general three-qubit input state.  This also invites for a detailed analysis on the hierarchy of nonlocality-breaking, steerability-breaking,  and entanglement-breaking quantum channels in the tripartite scenario. \bigskip    
	\section*{Acknowledgment}
	Authors thank Guruprasad Kar for the initial motivation of the work. SK acknowledges Yeong-Cherng Liang for the fruitful comments and the support from the National Science and Technology Council, Taiwan (Grants no.109-2112-M-006-010-MY3,110-2811-M-006-511,111-2811-M-006-013). JN’s work was supported by the Foundation for Polish Science within the “Quantum Optical Technologies” project carried out within the International Research Agendas programme cofinanced by the European Union under the European Regional Development Fund. AKP acknowledges the support from the research grant DST/ICPS/QuEST/2019/4.
%
\end{document}